\documentclass{amsart}

\usepackage{graphicx}
\usepackage{xcolor}

\usepackage[sorting=none]{biblatex}
\allowdisplaybreaks[1]

\newtheorem{theorem}{Theorem}[section]

\newtheorem{prop}[theorem]{Proposition}
\newtheorem{corollary}[theorem]{Corollary}

\theoremstyle{definition}
\newtheorem{definition}[theorem]{Definition}
\newtheorem{example}[theorem]{Example}

\theoremstyle{remark}

\numberwithin{equation}{section}

\newcommand{\Det}{\operatorname{Det}}

\DeclareMathOperator{\tr}{tr}
\DeclareMathOperator{\can}{\nabla\mathbf{Det}} 

\begin{document}

\title{Differentiating the pseudo determinant}

\author{Andrew Holbrook}
\address{Department of Statistics, UC Irvine}
\email{aholbroo@uci.edu}
\thanks{This work was supported by the UC Irvine Dean's Dissertation Fellowship.  I thank Professor Oliver Knill for his encouragement and helpful discussion, and I am grateful to an anonymous reviewer for further helpful recommendations.}



\date{}



\begin{abstract}
A class of derivatives is defined for the pseudo determinant $\Det(A)$ of a Hermitian matrix $A$.  This class is shown to be non-empty and to have a unique, canonical member $\can(A)=\Det(A)A^+$, where $A^+$ is the Moore-Penrose pseudo inverse.  The classic identity for the gradient of the determinant is thus reproduced. Examples are provided, including the maximum likelihood problem for the rank-deficient covariance matrix of the degenerate multivariate Gaussian distribution.
\end{abstract}

\maketitle

\section{Introduction}
We derive the class of derivatives of the pseudo determinant with respect to Hermitian matrices, placing an emphasis on understanding the forms taken by this class and their relationship to established results in linear algebra. In particular, care must be taken to address the discontinuous nature of the pseudo derivative. The contributions in this paper are primarily of a linear algebraic nature but are well motivated in fields of application.

The pseudo determinant arises in graph theory within Kirchoff's matrix tree theorem \cite{knill2014cauchy} and in statistics, in the definition of the degenerate Gaussian distribution. The degenerate Gaussian has been useful in image segmentation \cite{yang2008unsupervised}, communications \cite{castaneda2014estimation}, and as the asymptotic distribution for multinomial samples \cite{wasserman2013all}. Despite these appearances, knowledge of how to differentiate the distribution's density function is conspicuously absent from the literature, and---since differentiation is often essential for maximization---the lack of this knowledge is a plausible barrier to the distribution's wider use.

Specifically, to obtain the maximum likelihood (ML) estimator for the singular covariance matrix of the degenerate Gaussian, one must be able to calculate the derivative of the log likelihood and hence the pseudo determinant of the covariance. Although \cite{anderson1985maximum} firmly establishes the subject of ML estimation for multivariate Gaussians, the authors never directly address singular covariance estimation.  This problem is explored in Section 3.  In Section 2, the pseudo determinant is introduced, and its derivative with respect to Hermitian matrices is derived.

\section{The canonical derivative}
We begin by introducing the pseudo determinant both as a product of eigenvalues and as a limiting form.
\begin{definition}
The pseudo determinant $\Det$ of a square matrix $A$ is defined as the product of its non-zero eigenvalues. If a matrix has no non-zero eigenvalues, then we say $\Det(0)=1$.
\end{definition}
See \cite{knill2014cauchy} for an equivalent definition of the pseudo determinant in terms of the characteristic polynomial. In deriving its derivative, it will be useful to write the pseudo determinant as a limit. 
\begin{prop}
If $A$ is an $n\times n$ matrix of rank $k$, then $\Det(A)$ is the limit
\begin{align}
 \Det(A) =  \lim_{\delta \rightarrow 0} \frac{\det(A+\delta I)}{\delta^{n-k}} 
\end{align}
for $\det(\cdot)$ the regular determinant. 
\end{prop}
Whereas this result is known \cite{minka1998inferring}, we were unable to find its proof, so it is given here in the spirit of completeness.
\begin{proof}
We use the identity
\begin{align}
    \det(X + ZYZ^*) = \det(Y^{-1} + Z^*X^{-1}Z)\,\det(Y)\,\det(X) \, .
\end{align}
Replacing $X$ with $k\,I_n$ and letting $A=U\Lambda U^*=ZYZ^*$, we have
\begin{align}
    \lim_{\delta \rightarrow 0} \frac{\det(A+\delta I)}{\delta^{n-k}} &= \lim_{\delta \rightarrow 0} \frac{k^n}{k^{n-r}} \det(\Lambda^{-1}+\frac{1}{k}I_r)\,\det(\Lambda) \\ \nonumber
    &= \Det(A) \lim_{k \rightarrow 0} k^r \det(\Lambda^{-1}+\frac{1}{k}I_r) \\ \nonumber
    &= \Det(A) \lim_{k \rightarrow 0}  \det(k\Lambda^{-1}+I_r) \\ \nonumber
    &= \Det(A) \, .
\end{align}
\end{proof}
Next, we define the Moore-Penrose pseudo inverse \cite{golub1973differentiation}, an important object involved in the derivative of the pseudo determinant.
\begin{definition}
The pseudo inverse $A^+$ of a matrix $A$  is also  defined in terms of a limit:
\begin{align}
A^+ = \lim_{\delta \rightarrow 0} A^*(AA^* + \delta I)^{-1} = \lim_{\delta \rightarrow 0} (A^*A + \delta I)^{-1}A^* \, .
\end{align}
$A^+$ exists in general and is unique.  It may also be defined as the matrix satisfying all the following criteria:
\begin{enumerate}
    \item $AA^+A=A$
    \item $A^+AA^+ = A^+$
    \item $(AA^+)^*=AA^+$
    \item $(A^+A)^*=A^+A$
\end{enumerate}
For Hermitian matrices, the pseudo inverse is obtained by inverting the matrix eigenvalues.
\end{definition}
As is the case for the pseudo inverse \cite{golub1973differentiation}, the pseudo determinant is discontinuous. For an example, consider the two matrices
\begin{align}
    A = \left( \begin{array}{cc}
    1   & 0  \\
        0 & 0
    \end{array} \right), \quad \mbox{and} \quad \quad B_j = \left( \begin{array}{cc}
    0   & 0  \\
        0 & j
    \end{array} \right) \, .
\end{align}
Note that $\Det(A)=1$ and $\Det(A+B_j)=j$, but that
\begin{align}
    \lim_{j\rightarrow 0}\Det(A+B_j)=0\neq 1 = \Det(\lim_{j\rightarrow0} A+B_j) \, .
\end{align}
As one might gather from this example, the pseudo determinant is discontinuous between sets of matrices of differing ranks. This discontinuity will effect the way we define the derivative of the pseudo determinant. We now turn to deriving this derivative.

For matrix $A$ in the space of $n\times n$ matrices $M_{n\times n}$, the matrix derivative of a function $h:M_{n\times n} \rightarrow \mathbb{R}$ is given by the matrix $\nabla h(A)$ satisfying
\begin{align}\label{directderiv}
   \nabla_Bh(A)= \tr\big( B \nabla h(A)\big)= \lim_{\tau \rightarrow 0} \frac{h(A+\tau B)-h(A)}{\tau} \, 
\end{align}
for any matrix $B\in M_{n\times n}$,  where $\nabla_Bh(A)$ is the directional derivative. We use the directional derivative to define the derivative of the pseudo determinant, but, on account of the discontinuity of the pseudo determinant, we must restrict the directions $B$ in which the directional derivative is defined. For this reason, we may define the derivative at a point only in certain directions and must modify the common definition of the directional derivative.
\begin{definition}
(Definition 1) For a matrix $A$ in the space of Hermitian $n\times n$, rank $k$ matrices $M^k_{n\times n}$, the directional derivative $\nabla_B \Det(A)$ of the pseudo determinant $\Det:M_{n\times n} \rightarrow \mathbb{R}$ is defined in directions $B\in M_{n\times n}^k$ \emph{that share the same kernel as} $A$, i.e. for which $Ker(A)=Ker(B)$. Then the derivative $\nabla \Det(A)$ is given by any matrix satisfying
\begin{align}\label{directderiv2}
   \nabla_B \Det(A)= \tr\big( B \nabla \Det(A)\big)= \lim_{\tau \rightarrow 0} \frac{\Det(A+\tau B)-\Det(A)}{\tau} \, .
\end{align}
\end{definition}
Note that, according to this definition, $\nabla \Det(A)$ is not unique, since it can take on different values along the kernel of $B$. This non-uniqueness can also be seen using the following class equations for the class of derivatives $\nabla \Det(A)$   of the pseudo determinant at a matrix $A$. 
\begin{definition}\label{classEq}
(Definition 2) A derivative of the pseudo determinant at a point $A\in M^k_{n\times n}$ is any non-zero matrix $\nabla \Det(A) \in M^{k}_{n\times n}$ satisfying
\begin{align}
    A\, \nabla \Det(A)  &= A\, A^+ \Det(A)  \\
   \nabla \Det(A) \, A &=  A^+A\, \Det(A)\, .
\end{align}
\end{definition}
 We demonstrate that this is a natural definition using the facts that $A(A^2)^+=A^+$ and $(A^2)^+A=A^+$ for any Hermitian $A$ and assuming one may interchange limits:
\begin{align}
  A^{1/2}  \nabla \Det(A) &= A^{1/2} \nabla \, \lim_{\delta \rightarrow 0} \frac{\det(A+\delta I)}{\delta^{n-k}} \\ \nonumber
    &=  A^{1/2} \lim_{\delta \rightarrow 0} \frac{1}{\delta^{n-k}}\, \nabla  \det(A+\delta I) \\ \nonumber
    &=  \Det(A)\, \lim_{\delta \rightarrow 0}  A^{1/2} (A+\delta I)^{-1} \\ \nonumber
    &=  \Det(A)\,   (A^{1/2})^+ \\ \nonumber
    &=  \Det(A)\,   A^{1/2}A^+  \, .
\end{align}
Multiplying both sides by $A^{1/2}$ and rearranging gives the first class equation.  The derivation of the second equation is symmetric. 
We illustrate the preceding definitions---and that they do not define unique derivatives---with a few examples. 
\begin{example}
Consider the $2\times 2$ matrix
\begin{align}
    A = \left( \begin{array}{cc}
    a    & 0  \\
        0 & 0
    \end{array} \right) \, .
\end{align}
It is clear that $\Det(A)=a$ and $A^+$ is obtained by taking the reciprocal of the first element of $A$. The above result renders
\begin{align}
    A\nabla \Det(A) = a\,AA^+ = \left( \begin{array}{cc}
    a   & 0  \\
        0 & 0
    \end{array} \right)= a\, A^+A= A\nabla \Det(A)  \, .
\end{align}
Note that multiple matrices solve this equation.  Two examples are the identity and the matrix $A/a$.
\end{example}
\begin{example}
Now consider the  $2\times 2$ matrix pair
\begin{align}
    A = \left( \begin{array}{cc}
    1   & 1  \\
        1 & 1
    \end{array} \right), \, \quad A^+ = \left( \begin{array}{cc}
    .25   & .25  \\
        .25 & .25
    \end{array} \right) \, .
\end{align}
One can check that $\Det(A)=2$. Then it follows from the result that
\begin{align}\label{Ex2}
    A\nabla \Det(A) = 2\, AA^+ = 2\, \frac{1}{2} A = A = \dots = \nabla \Det(A) A \, .
\end{align}
Again, multiple matrices satisfy Equation \eqref{Ex2}: take for example
\begin{align}
     \left( \begin{array}{cc}
    1    & 0  \\
        0 & 1
    \end{array} \right)  \quad \mbox{and} \quad 
    \left( \begin{array}{cc}
    .5   & .5  \\
        .5 & .5
    \end{array} \right) \, .
\end{align}
\end{example}
It turns out that the matrix $A$ in the class equations of Definition \eqref{classEq} may be replaced by any Hermitian $B$ such that $Ker(B)=Ker(A)$.  This is easily shown using the fact that $BA^+A=B=BAA^+=B=A^+AB=B=AA^+B$ for any such $B$.
\begin{prop}
The derivative of the pseudo determinant is any matrix $\nabla \Det(A)$ satisfying the equations
\begin{align}
    B\, \nabla \Det(A)  &= B\, A^+ \Det(A)  \\
   \nabla \Det(A) \, B &=  A^+B\, \Det(A)\, ,
\end{align}
for any matrix $B$ for which $Ker(B)=Ker(A)$.
\end{prop}
This result may be combined with the directional derivative based definition of $\nabla \Det(A)$.
\begin{prop}
The derivative of the pseudo determinant is any matrix $\nabla \Det(A)$ satisfying the equations
\begin{align}\label{prop2}
    \tr\big(B\nabla \Det(A)\big) &= \Det(A) \tr(BA^+)  .
\end{align}
for any matrix $B$ for which $Ker(B)=Ker(A)$.
\end{prop}
In practice, one may obtain the canonical element $\can(A)$ of class $\nabla \Det(A)$ directly from a corollary to the following Pythagorian theorem.
\begin{theorem}
(Knill 2014 \cite{knill2014cauchy}) For Hermitian $A$ of rank $k$,
\begin{align}
    \Det^2(A)=\Det(A^2) = \sum_P \det\nolimits^2 (A_P) \,
\end{align}
where $P$ indexes all $k\times k$ minors of $A$ satisfying $\det(A_P)\neq0$.
\end{theorem}
As a corollary, the canonical gradient $\can$ is directly obtainable.
\begin{corollary}\label{gradDetCor}
For Hermitian $A$ of rank $k$, one has
\begin{align}\label{gradDet}
    \nabla \Det(A) = \frac{1}{\Det(A)}  \sum_P \det\nolimits^2(A_P)A^{-1}_P = \frac{\sum_P \det^2(A_P)A^{-1}_P}{\sqrt{\sum_P \det^2(A_P)}}:=\can(A) \, .
\end{align}
\end{corollary}
This $\can(A)$ satisfies the class equations as well as Equation \eqref{prop2}. Before proving this claim, we illustrate by revisiting our examples.
\begin{example}
We again consider matrix
\begin{align}
    A = \left( \begin{array}{cc}
    a    & 0  \\
        0 & 0
    \end{array} \right) \, .
\end{align}
This time we use Formula \eqref{gradDet}.  Here, the rank $k$ minors are simply the elements of $A$.  Since only the first element is non-zero, we have
\begin{align}
    \can(A) = \frac{\det^2(A_{11})\,A_{11}^{-1}}{\Det(A)} = \frac{a^2}{a} \left( \begin{array}{cc}
    a^{-1}    & 0  \\
        0 & 0
    \end{array} \right) = \left( \begin{array}{cc}
    1    & 0  \\
        0 & 0
    \end{array} \right)\, .
\end{align}
This, of course, agrees with the original example.
\end{example}

\begin{example}
Again, consider the matrix
\begin{align}
    A = \left( \begin{array}{cc}
    1   & 1  \\
        1 & 1
    \end{array} \right)\, .
\end{align}
The gradient of the pseudo determinant may be found using Formula \eqref{gradDet}:
\begin{align}
    \can(A) = \frac{1}{2} \sum_{ij} \det\nolimits^2(A_{ij}) A_{ij}^{-1} = \frac{1}{2} A \, .
\end{align}
The reader may check that
\begin{align}
    A\can(A)=\frac{1}{2}A^2= A = \dots = \mathbf{\nabla \Det(A)} A \, ,
\end{align}
as expected from Equation \eqref{Ex2}.
\end{example}
The above examples suggest that $\can(A)$ should satisfy the class equations in general.  To show this, we first cite a result.
\begin{theorem}\label{berg}
(Berg 1986 \cite{berg1986three}) The pseudo inverse of a Hermitian, rank $k$ matrix $A$ takes the following form:
\begin{align}
    A^+ = \frac{\sum_P \det^2(A_P)A^{-1}_P}{\Det^2(A)} =\frac{\sum_P \det^2(A_P)A^{-1}_P}{\sum_P \det^2(A_P)}\,.
\end{align}
\end{theorem}
\begin{theorem}
\begin{align}\label{canonicalFormula}
    \can(A)=\Det(A)A^+ \,
\end{align}
 Thus $\can(A)$ satisfies the class equations and belongs to the equivalence class $\nabla \Det(A)$. Moreover, $\can(A)$ is the unique member of the equivalence class that has the same kernel as $A$. In this sense, it may be considered the canonical gradient of the pseudo determinant.
\end{theorem}
\begin{proof}
That $\can(A)=\Det(A)A^+$ is a simple result of Corollary \ref{gradDetCor} and Theorem \ref{berg}. As a result, it immediately satisfies the two propositions as well.

We now consider the uniqueness claim.  In general, $A:Ker(A)^{\perp}\rightarrow Im(A)$ is an isomorphism, and $A:Im(A)\rightarrow Ker(A)^{\perp}$ is its inverse. Since $A$ is Hermitian, $Ker(A)\oplus Im(A)=\mathbb{C}^n$, and so $A:Im(A)\rightarrow Im(A)$, $A^+:Im(A)\rightarrow Im(A)$ is the isomorphism pair. Clearly $Ker(A)=Ker(A^+)$, and so $Ker\big(\can(A)\big)=Ker(A)$. 

We proceed by contradiction. Suppose that there exists another matrix $B\neq \can(A)$ satisfying $Ker(A)=Ker(B)$ and
\begin{align}
    AB  &= A\, A^+ \Det(A)  \\ \nonumber
   BA &=  A^+A\, \Det(A)\, .
\end{align}
Since $B\neq A$, there exists at least one element $y\in \mathbb{C}^n$ such that $By\neq \can(A)y$.  Since $\mathbb{C}^n= Im(A) \oplus Ker(A)$, we may consider $y$ in each subspace separately. If $y \in Ker(A)$, then $By=0=\can(A)y$.  Therefore $y$ must be an element of $Im(A)$.  Then, 
\begin{align}
    (B-\can(A))y &= (B-\can(A)) (AA^+) y \\ \nonumber
                      &= (BA - \can(A)A)A^+ y \\ \nonumber
                      &= (A^+A\,\Det(A) - A^+A\,\Det(A))A^{+} y \\ \nonumber
                      &= 0 \, .
\end{align}
Then $By=\can(A)y$, thus establishing a contradiction.
\end{proof}

We round out this section with a few examples demonstrating applications of Formula \eqref{canonicalFormula}.
\begin{example}
Let $A$ be the constant, $n\times n$ matrix satisfying $A_{ij}=1, \forall i,j =1,\dots,n$. Then it is true that 
\begin{align}
    \Det(A) = n\,, \quad \mbox{and} \quad A^+ = \frac{1}{n^2} A \, .
\end{align}
Hence
\begin{align}
    \can(A) = \Det(A) A^+ = \frac{1}{n} A \, .
\end{align}
\end{example}
\begin{example}
Let $A=0$ be the $n\times n$ zero matrix for arbitrary integer $n$. The reader can check that $A^+=A=0$ by observing the four criteria in the definition of the pseudo inverse.  Recall also that $\Det(0)=1$ for any square matrix with no non-zero eigenvalues.  It follows that
\begin{align}
    \can (A) = \Det(A) A^+ = A = 0 \, .
\end{align}
This basic result is more appealing using the shorthand $\can(0)=0$.
\end{example}

\begin{example}
Consider the projection-dilation matrix
\begin{align}
    A = \left( \begin{array}{cc}
    a^2   & ab  \\
        ab & b^2
    \end{array} \right)
\end{align}
that maps a point $v \in \mathbb{R}^2$ onto the line through the origin containing the unit vector $u=(a,b)^T/\sqrt{(a^2+b^2)}$ while scaling by $a^2+b^2$. The reader may check that
\begin{align}
    \Det(A)=a^2+b^2 \,, \quad \mbox{and} \quad A^+ = \frac{1}{(a^2+b^2)^2}A \, .
\end{align}
We thus obtain the intriguing result
\begin{align}
    \can(A) = \frac{1}{a^2+b^2}A = \frac{1}{a^2+b^2} \left( \begin{array}{cc}
    a^2   & ab  \\
        ab & b^2
    \end{array} \right) = \frac{1}{(a,b)(a,b)^T}(a,b)^T(a,b)\, ,
\end{align}
where the last form is meant to make clear that the result is the projection onto the subspace spanned by $(a,b)^T$. 
\end{example}

The previous example touches on graph theory if we let $(a,b)=(\sqrt{c},-\sqrt{c})$.  

\begin{example}

Let $L$ denote the Laplacian $L=D-A$ of a weighted graph, where $A$ is the weighted adjacency matrix having zeros down the diagonal and off-diagonal elements $A_{ij}$ equal to the value associated with the edge connecting nodes $i$ and $j$. The matrix $D$ is diagonal and has elements satisfying $D_{ij}=\sum_iA_{ij}=\sum_jA_{ij}$.

In the special case of a connected, two node graph with edge value $c$, the Laplacian is
\begin{align}
    L =  \left( \begin{array}{cc}
    c   & 0  \\
    0    & c
    \end{array} \right) -  \left( \begin{array}{cc}
    0   & c  \\
    c    & 0
    \end{array} \right)  =  c \cdot \left( \begin{array}{cc}
    1   & -1  \\
    -1    & 1
    \end{array} \right) \, .
\end{align}
Noting that $L$ is a projection-dilation matrix (see prior example), we get
\begin{align}
        \Det(L)=\sqrt{c}^2+(-\sqrt{c})^2=2c \,, \quad \mbox{and} \quad L^+ = \frac{1}{4c^2}L \, ,
\end{align}
and thus, by Formula \eqref{canonicalFormula}, 
\begin{align}
     \can(L) = \frac{2c}{4c^2}L = \frac{1}{2c}L= \frac{1}{2}\left( \begin{array}{cc}
    1   & -1  \\
        -1 & 1
    \end{array} \right).
\end{align}
The last term is half the Laplacian associated to the simple, \emph{unweighted} graph obtained by removing the weight $c$. Hence, $\can(L)$ takes graph connectivity into account but not scale.

\end{example}

\subsection{The matrix differential}
 When obtaining matrix derivatives, it is often easiest to calculate the matrix differential $dA$  and then relate back to the gradient using the formula \cite{magnus1988matrix}
\begin{align}\label{gradTOdiff}
    dh(A)=\tr\big( (dA)\,G\big) \iff \nabla h(A) = G \, .
\end{align}
Combining this identity with directional derivative Formula \eqref{directderiv}, we see that $Ker(dA)$ must equal $Ker(A)$ for the special case of the derivative of the pseudo determinant.  It follows that the matrix differential of the pseudo determinant is
\begin{align}\label{matrixdiff}
    d \Det(A) = \Det(A)\, \tr\big(A^+ (dA)\big) \, ,
\end{align}
where we are implicitly selecting for the canonical gradient $\can(A)$ in order to satisfy $Ker(dA)=Ker(A)$.  Equation \eqref{matrixdiff} may also be derived directly using the spectral decomposition $A=U\Lambda U^*=\sum_{j=1}^k \lambda_j \,u_ju_j^*$ for rank $k$, Hermitian $A$. The differential of an eigenvalue of a Hermitian matrix $A$ may be written in terms of the matrix differential itself \cite{magnus1988matrix}:
\begin{align}\label{eigendiff}
    d\lambda = \tr\big(uu^* \, (dA)\big) \, .
\end{align}
\begin{theorem}\label{diffTheorem}
The matrix differential of the pseudo determinant of Hermitian $A\in M^k_{n\times n}$ is
\begin{align}
    d \Det(A) = \Det(A)\, \tr\big(A^+ (dA)\big) \, .
\end{align}
\end{theorem}
\begin{proof}
 The result is proven directly using Formula \eqref{eigendiff}.
 \begin{align}
      d \Det(A) &= d \prod_{j=1}^k \lambda_j \\ \nonumber
               &= \sum_{j=1}^k d\lambda_j \, \prod_{i\neq j} \lambda_i \\ \nonumber
               &= \sum_{j=1}^k \tr\big(u_ju_j^*\,(dA)\big) \, \prod_{i\neq j} \lambda_i \\ \nonumber
               &= \sum_{j=1}^k \tr\big(\frac{1}{\lambda_j}u_ju_j^*\,(dA)\big) \, \prod_{i=1}^k \lambda_i \\ \nonumber
               &= \Det(A) \sum_{j=1}^k \tr\big(\frac{1}{\lambda_j}u_ju_j^*\,(dA)\big)  \\ \nonumber
               &= \Det(A) \, \tr\big(\sum_{j=1}^k\frac{1}{\lambda_j}u_ju_j^*\,(dA)\big)  \\ \nonumber
               &= \Det(A) \, \tr\big(A^+(dA)\big)  
 \end{align}
\end{proof}
The reader should note that Theorem \ref{diffTheorem} could also be used to derive the canonical gradient $\can(A)$ via Formula \eqref{gradTOdiff}.

\section{An example from statistics}
We now derive the maximum likelihood estimator (MLE) for the singular covariance of the degenerate multivariate Gaussian distribution. Thus, this section may be considered an extension of the results found in \cite{anderson1985maximum}.  The MLE may be incorporated into more advanced statistical algorithms such as expectation maximization for image segmentation \cite{yang2008unsupervised}. The formulas derived in the following are also potentially useful in a Hamiltonian Monte Carlo algorithm for Bayesian inference over reduced-rank covariance matrices (cf. \cite{holbrook2018geodesic}).

Let $x_1,\dots,x_N$ follow a degenerate Gaussian distribution with mean $\mu$ and singular covariance $\Sigma$. The probability density function of such a random variable $x_i$ is given by
\begin{align}
    f(x_i;\mu,\Sigma) = \Det(2\pi\Sigma)^{-1/2}\exp \big(-\frac{1}{2}(x_i-\mu)^T\Sigma^+(x_i-\mu)\big) \, .
\end{align}
Assuming that $\mu$ is known, the log-likelihood $\ell(\Sigma)$ of $\Sigma$ is proportional to
\begin{align}
  - N \log \big( \Det(\Sigma)\big) - \sum_{i=1}^N(x_i-\mu)^T\Sigma^+(x_i-\mu) = -N \log \big( \Det(\Sigma)\big) - \tr\big( \Sigma^+R\big) \, ,
\end{align}
where $R$ is the matrix of residuals.

To obtain the MLE $\hat{\Sigma}$, we obtain the gradient of $\ell(\Sigma)$ and set it to zero, just as in the case of a full-rank covariance matrix. To calculate the second term in the log-likelihood, we need the formula for the matrix differential of the pseudo inverse \cite{golub1973differentiation}:
\begin{align}
    d\Sigma^+ = -\Sigma^+ (d\Sigma)\Sigma^+ + \Sigma^+\Sigma^+(d\Sigma)(I-\Sigma\Sigma^+)+(I-\Sigma^+\Sigma)(d\Sigma)\Sigma^+\Sigma^+ \, . 
\end{align}
It follows that
\begin{align}
    d \ell(\Sigma) &= - N\tr\big(\Sigma^+ (d\Sigma)\big) + \tr\big(\Sigma^+ (d\Sigma)\Sigma^+R\big) \\ \nonumber
    &\quad - \tr\big(\Sigma^+\Sigma^+(d\Sigma)(I-\Sigma\Sigma^+)R\big) - \tr\big((I-\Sigma^+\Sigma)(d\Sigma)\Sigma^+\Sigma^+R\big) \\ \nonumber 
    &= -N \tr\big(\Sigma^+ (d\Sigma)\big) + \tr\big(\Sigma^+R\Sigma^+(d\Sigma)\big) \\ \nonumber
    &\quad - \tr\big((I-\Sigma\Sigma^+)R\Sigma^+\Sigma^+(d\Sigma)\big) - \tr\big(\Sigma^+\Sigma^+R(I-\Sigma^+\Sigma)(d\Sigma)\big) \, .
\end{align}
Setting $d\ell(\hat{\Sigma})=0$ and applying Formula $\eqref{gradTOdiff}$, we get
\begin{align}
   N \hat{\Sigma}^+ = \hat{\Sigma}^+R\hat{\Sigma}^+ - (I-\hat{\Sigma}\hat{\Sigma}^+)R\hat{\Sigma}^+\hat{\Sigma}^+ - \hat{\Sigma}^+\hat{\Sigma}^+R(I-\hat{\Sigma}^+\hat{\Sigma}) \, ,
\end{align}
and multiplying both sides by $\hat{\Sigma}$ on the right and the left gives
\begin{align}
    N\hat{\Sigma} &= \hat{\Sigma}\hat{\Sigma}^+R\hat{\Sigma}^+\hat{\Sigma}-\hat{\Sigma} (I-\hat{\Sigma}\hat{\Sigma}^+)R\hat{\Sigma}^+\hat{\Sigma}^+ \hat{\Sigma} - \hat{\Sigma} \hat{\Sigma}^+\hat{\Sigma}^+R(I-\hat{\Sigma}^+\hat{\Sigma})\hat{\Sigma}  \\ \nonumber
    &= \hat{\Sigma}\hat{\Sigma}^+R\hat{\Sigma}^+\hat{\Sigma} \, . 
\end{align}
This last line follows because the matrices $\Sigma\Sigma^+$ and $\Sigma^+\Sigma$ are projections onto the range of $\Sigma$ and $\Sigma^+$, and therefore  $(I-\Sigma^+\Sigma)$ and $(I-\Sigma\Sigma^+)$ annihilate $\Sigma$. For the same reason, \emph{if we are willing to assume that} $Ker(R)=Ker(\Sigma)$, this last equation is solved by
\begin{align}
    \hat{\Sigma}=\frac{1}{N}\hat{\Sigma}\hat{\Sigma}^+R\hat{\Sigma}^+\hat{\Sigma} = \frac{1}{N}R \, .
\end{align}
\emph{Thus only with that key assumption are we able to reproduce the classical result for full rank} $\Sigma$. If we are not willing to make this assumption, i.e. if we have prior belief that, or have set up our model in such a way that, the range of $\Sigma$ is a predetermined subspace, then the above equation may be written
\begin{align}
    \hat{\Sigma}=\frac{1}{N}\hat{\Sigma}\hat{\Sigma}^+R\hat{\Sigma}^+\hat{\Sigma}=\hat{\Sigma}=\frac{1}{N}\Sigma\Sigma^+R\Sigma^+\Sigma \,.
\end{align}
Then $\hat{\Sigma}$ is precisely the projection of the residual matrix $R/N$ onto the range of $\Sigma$.
\printbibliography

\end{document}